\documentclass[10pt,twocolumn]{article}

\usepackage{amsmath,amssymb}
\usepackage{amsthm}
\usepackage{times}
\usepackage{hyperref}
\usepackage{enumitem}
\usepackage{geometry}

\theoremstyle{plain}
\newtheorem{theorem}{Theorem}[section]
\newtheorem{lemma}[theorem]{Lemma}

\theoremstyle{definition}
\newtheorem{definition}[theorem]{Definition}

\theoremstyle{remark}
\newtheorem{remark}[theorem]{Remark}

\title{Two Linear Passes Are Necessary for Sum-Exclude-Self Under Sublinear Space}
\author{Andrew Au\\[0.2em]
\small Independent Researcher\\[0.2em]
\small Corresponding author: \texttt{cshung@gmail.com}}
\date{\vspace{-1.5em}}

\begin{document}

\maketitle

\begin{abstract}
We prove that any algorithm computing the sum-exclude-self of an unsigned $d$-bit integer array of length $n$ under sublinear space must perform two linear passes over the input. More precisely, the algorithm must read at least $n-1$ input elements before any output cell receives its final value, and at least
$n - \lfloor t/d \rfloor$ additional elements thereafter, where $t = o(nd)$ bits is the working memory size. This gives a total of $2n - 1 - \lfloor t/d \rfloor$
element reads. A trivial modification of the standard two-pass algorithm achieves this bound
exactly for all practical input sizes. The proof uses this toy problem as a worked example
to demonstrate the choke-point technique for proving sublinear-space lower bounds.
\end{abstract}


\section{Introduction}

The \emph{sum-exclude-self} problem asks: given an array $\mathit{In}$ of $n$ $d$-bit unsigned
integers from the domain $\{0, \ldots, 2^d-1\}^n$, compute an output array $\mathit{Out} \in \mathbb{Z}^n$
such that
\[
  \mathit{Out}[i] = \sum_{j \neq i} \mathit{In}[j]
\]
for all $0 \le i < n$. Each input element $\mathit{In}[i]$ occupies exactly $d$ bits of storage,
but intermediate values (such as the total sum $S$) and output values are unbounded integers
in $\mathbb{Z}$.

The standard two-pass algorithm works as follows. In the first pass, compute the total sum
$S = \sum_{j=0}^{n-1} \mathit{In}[j]$. In the second pass, for each $i$, output
$\mathit{Out}[i] = S - \mathit{In}[i]$. Every input element is read exactly twice.

Our goal is to prove that any algorithm solving this problem with sublinear working
memory $t = o(nd)$ bits must perform two linear passes over the input: it must read
at least $n-1$ input elements before any output cell receives its final value, and must read at least $n - \lfloor t/d \rfloor$ elements
thereafter. The standard two-pass algorithm's $2n$ reads nearly matches this lower bound;
for $d=32$ and $n \le 2^{31}$, the gap is just $2$ reads.

Rather than claiming this problem is practically difficult, we use it as a toy example that
isolates the choke-point technique---passing the identity function through an information
bottleneck created by sublinear memory---without the machinery being obscured by problem
complexity.

\subsection{Model of Computation}

The algorithm reads input elements from $\mathit{In}$ and writes output elements to $\mathit{Out}$,
maintaining at most $t = o(nd)$ bits of working memory between any read or write operation.
The input array $\mathit{In}$ is \emph{read-only}.
The output array $\mathit{Out}$ is \emph{write-only}: the algorithm may write to any output cell
at any time (and may overwrite previously written values), but may never read from $\mathit{Out}$.
This prevents the algorithm from using output storage as additional working memory.
Arithmetic is performed over the integers $\mathbb{Z}$ with no overflow.
We assume $n \ge 2$ throughout; the case $n = 1$ is trivial since
$\mathit{Out}[0] = 0$ for every input.

\textbf{The $d$-bit input restriction is essential}: The input domain $\{0, \ldots, 2^d-1\}^n$
carries $nd$ bits of information ($2^{nd}$ distinct inputs). This finite, measurable domain
is what makes the information-theoretic lower bound possible. Without this boundedness,
the input domain $\mathbb{Z}^n$ is infinite and carries no finite bit measure, so
no analogous counting argument applies. The $d$-bit 
restriction applies only to input storage, not to intermediate values or outputs.

\section{First Pass Lower Bound}

\begin{definition}
For each output cell $\mathit{Out}[i]$, the \emph{final write} is the last write to that cell
during the algorithm's execution.  Since the algorithm is correct, each final write produces the
correct output value.

The \emph{first pass} is the time interval from the start of the algorithm until immediately
before the earliest final write to any output cell.
\end{definition}

\begin{lemma}
\label{lem:firstpass}
During the first pass, the algorithm must read at least $n-1$ elements of $\mathit{In}$.
\end{lemma}

\begin{proof}
Fix any deterministic algorithm $A$.
Suppose, for contradiction, that the earliest final write occurs after $A$ has read at most $n-2$ distinct input positions.  
Let $S$ be the set of indices read, so $|S| \le n-2$, and let $i$ be the index receiving this final write.
Since $\mathit{Out}$ is write-only, $A$'s entire state at this moment depends only on $\{\mathit{In}[j] : j \in S\}$.

Since $|S| \le n-2$, at least two indices are unread. At most one of them equals $i$, so there exists an unread index $p \notin S$ with $p \neq i$.
Changing $\mathit{In}[p]$ does not affect $A$'s state (since $p$ is unread), but changes the correct value $\mathit{Out}[i] = \sum_{j \neq i} \mathit{In}[j]$ (since $p \neq i$).
Thus $A$'s final write to $\mathit{Out}[i]$ cannot be correct for both inputs, contradicting correctness.
\end{proof}

\begin{remark}
The $n-1$ bound is tight. An algorithm could read all elements except $\mathit{In}[0]$,
then make its earliest final write $\mathit{Out}[0] = \mathit{In}[1] + \cdots + \mathit{In}[n-1]$,
having read exactly $n-1$ elements.
\end{remark}

At the end of the first pass, the working memory contains at most $t = o(nd)$ bits summarizing
everything known about $\mathit{In}$.

\section{Second Pass Lower Bound}

\begin{definition}
The \emph{second pass} is the interval from the earliest final write to the end of execution.
\end{definition}

\subsection{Setup}

Let $T$ denote the space of all possible working memory states at the start of the second pass;
each state is at most $t$ bits, so $|T| \le 2^{t}$.
Let $u$ be the maximum number of bits read from $\mathit{In}$ during the second pass
over all inputs, and encode each input's second-pass read transcript as a $u$-bit string
(padding shorter transcripts with zeros), giving a set $U$ with $|U| \le 2^{u}$.
The algorithm computes the entire $\mathit{Out} \in \mathbb{Z}^n$ using only information from $T \times U$,
so there exists a function
\[
  f : T \times U \to \mathbb{Z}^n
\]
mapping the combined information through the choke point to the output space.

The total number of distinct values passing through this choke point is at most $|T \times U| \le 2^{t+u}$.

\subsection{Reconstruction}

We now show that any valid output $\mathit{Out}$ (produced from some input in $\{0,\ldots,2^d-1\}^n$) 
explicitly determines the original input $\mathit{In}$. This enables the identity map argument.

\begin{lemma}
\label{lem:explicit}
For $n \ge 2$, given valid $\mathit{Out}$ produced from input in $\{0,\ldots,2^d-1\}^n$, 
we can explicitly recover the original $\mathit{In}$.
\end{lemma}

\begin{proof}
First compute the sum of all output values:
\[
S_{\mathit{Out}} = \sum_{i=0}^{n-1} \mathit{Out}[i].
\]
Each $\mathit{Out}[i] = \sum_{j\neq i} \mathit{In}[j]$, so expanding the sum gives:
\begin{align*}
S_{\mathit{Out}} &= \sum_{i=0}^{n-1} \sum_{j\neq i} \mathit{In}[j] \\
&= \sum_{i=0}^{n-1} \left( \sum_{j=0}^{n-1} \mathit{In}[j] - \mathit{In}[i] \right) \\
&= \sum_{i=0}^{n-1} \left( S - \mathit{In}[i] \right) \\
&= nS - \sum_{i=0}^{n-1} \mathit{In}[i] = nS - S = (n-1)S,
\end{align*}
where $S = \sum_{j=0}^{n-1} \mathit{In}[j]$ is the total input sum. 

Thus $S = S_{\mathit{Out}} / (n-1)$. Each input element recovers as
\[
\mathit{In}[i] = S - \mathit{Out}[i].
\]
The division by $n-1$ is exact because $S_{\mathit{Out}}$ came from a valid input 
(this need not hold for arbitrary elements of $\mathbb{Z}^n$).
\end{proof}

With explicit reconstruction established, the identity map on input space factors through 
the choke point as a composition of three maps:

\begin{align*}
&\{0,\ldots,2^d-1\}^n \\
&\quad \xrightarrow{\text{first pass state $+$ second pass reads}} T\times U \\
&\quad \xrightarrow{f} \mathbb{Z}^n \\
&\quad \xrightarrow{\text{Lemma~\ref{lem:explicit}}} \{0,\ldots,2^d-1\}^n.
\end{align*}

\subsection{Choke Point Argument}

With reconstruction established, we derive the information-theoretic lower bound.

\begin{theorem}
\label{thm:secondpass}
Any algorithm using $t$ bits of working memory must read at least $nd - t$ bits
from $\mathit{In}$ during the second pass, i.e., at least $n - \lfloor t/d \rfloor$ elements.
\end{theorem}

\begin{proof}
The identity map on input space factors through $T \times U$ via the composition
in the diagram above. The input domain $\{0,\ldots,2^d-1\}^n$ carries
$nd$ bits of information ($2^{nd}$ distinct inputs).

Consider what happens when this identity map---a function from a set of size $2^{nd}$
onto itself---factors through the choke point $T \times U$. The space $T \times U$ 
has at most $2^{t+u}$ possible values. A function that is one‑to‑one (hence the identity)
cannot factor through an intermediate set smaller than its domain. 

Thus we must have $|T \times U| \ge 2^{nd}$, so $t + u \ge nd$, or $u \ge nd - t$.

Finally, since the algorithm reads quantized $d$-bit elements from $\mathit{In}$, 
the number of elements read in the second pass is at least $\lceil u/d \rceil$. 
From $u \ge nd - t$ we derive
\[
\lceil u/d \rceil 
\ge \left\lceil \frac{nd - t}{d} \right\rceil 
= \left\lceil n - \frac{t}{d} \right\rceil.
\]
Since $n$ is an integer, we have
\[
\left\lceil n - \frac{t}{d} \right\rceil
= n + \left\lceil -\frac{t}{d} \right\rceil
= n - \left\lfloor \frac{t}{d} \right\rfloor.
\]
Hence the second pass must read at least $n - \lfloor t/d \rfloor$ elements of $\mathit{In}$.
\end{proof}

\begin{remark}
Since $t = o(nd)$ is sublinear, this is $n - o(n)$.
\end{remark}

\section{Tightness Analysis}

To evaluate our lower bound, we analyze the standard algorithm's space usage and read count, 
then optimize it, revealing a persistent $\left\lfloor \frac{\lceil \log_2 n \rceil}{d} \right\rfloor$ gap.

\subsection{Standard Algorithm vs Lower Bound}

The standard algorithm stores the total sum $S = \sum_j \mathit{In}[j] \le n(2^d - 1)$.
Since $n(2^d - 1) < n \cdot 2^d \le 2^{d + \lceil \log_2 n \rceil}$,
storing $S$ requires at most $d + \lceil \log_2 n \rceil$ bits, so $t = d + \lceil \log_2 n \rceil$ suffices.

Our lower bound requires
\[
n - \left\lfloor \frac{t}{d} \right\rfloor 
= n - \left\lfloor \frac{d + \lceil \log_2 n \rceil}{d} \right\rfloor
\]
second‑pass reads. Simplifying gives:
\begin{align*}
\left\lfloor \frac{d + \lceil \log_2 n \rceil}{d} \right\rfloor
&= \left\lfloor 1 + \frac{\lceil \log_2 n \rceil}{d} \right\rfloor \\
&= 1 + \left\lfloor \frac{\lceil \log_2 n \rceil}{d} \right\rfloor.
\end{align*}
Thus the bound predicts $n - 1 - \left\lfloor \frac{\lceil \log_2 n \rceil}{d} \right\rfloor$ reads. The standard algorithm reads all $n$ elements, leaving a gap of at least 1 read.

\subsection{Optimizing the Constant Gap}

Avoid the final read with a running counter:
\[
\begin{array}{l@{\quad}l}
C \gets S & \text{(sum of all unread elements)} \\
\text{for } i = 0 \text{ to } n-2: & \\
\quad \text{read } \mathit{In}[i] & \\
\quad \text{write } \mathit{Out}[i] = S - \mathit{In}[i] & \\
\quad C \gets C - \mathit{In}[i] & \\
\text{write } \mathit{Out}[n-1] = S - C & \text{(no read)}
\end{array}
\]
This reads $n-1$ elements using $d + \lceil \log_2 n \rceil$ bits.

\subsection{The Remaining Gap and Limitations}

The optimized algorithm still exceeds the bound by $\left\lfloor \frac{\lceil \log_2 n \rceil}{d} \right\rfloor$ reads (zero for $n \le 2^{d-1}$, e.g.\ $n \le 2^{31}$ when $d=32$, but asymptotically positive). This analysis assumes the sum $S$ is the optimal summary---we have not ruled out better summaries or summary-free approaches. The true exact minimum number of reads remains open.

\section{Discussion}
This toy problem demonstrates five key ingredients of the choke-point technique:

\begin{enumerate}[leftmargin=*]
\item \emph{Operational pass definitions} based on output boundaries rather than syntactic passes;
\item \emph{Finite $d$-bit input domain} enabling cardinality arguments while keeping arithmetic
      over $\mathbb{Z}$ to avoid overflow complications;
\item \emph{Explicit reconstruction} showing the output space can recover the input;
\item \emph{Choke point bit budget} forcing information from input to pass through limited memory;
\item \emph{Read quantization} converting the bit lower bound to an element lower bound.
\end{enumerate}

These ingredients generalize broadly to proving lower bounds for problems under limited working memory.

\section*{Acknowledgements}

This research did not receive any specific grant from funding agencies in the public, commercial, or not-for-profit sectors.

\end{document}